\documentclass{amsart}
\usepackage{amsmath,mathtools,amssymb}
  \usepackage{paralist}
 \usepackage[colorlinks=true]{hyperref}
\hypersetup{urlcolor=red, citecolor=blue}

\newtheorem{theorem}{Theorem}[section]

\newtheorem{lemma}[theorem]{Lemma}

\theoremstyle{definition}

\title[Internal equatorial waves of extreme form] 
      {Geophysical internal equatorial waves of extreme form}

\author[Tony Lyons]{}


 \email{tlyons@wit.ie}
\thanks{The paper is for the special theme: Mathematical Aspects of Physical Oceanography, organized by Adrian Constantin.}

\begin{document}
\maketitle

\centerline{\scshape Tony Lyons}
\medskip
{\footnotesize
 \centerline{Department of Computing \& Mathematics}
   \centerline{Waterford Institute of Technology}
   \centerline{Waterford, Ireland}
}

\begin{abstract}
The existence of internal geophysical waves of extreme form is confirmed and an explicit solution presented. The flow is confined to a layer lying above an eastward current  while the mean horizontal flow of the solutions is westward, thus incorporating flow reversal in the fluid.
\end{abstract}

\section{Introduction}
 Hydrodynamic waves of extreme form have played a prominent role in the mathematical study of water waves since the 19th century, when Stokes \cite{Stokes1880} first predicted the existence of the wave of greatest height. Extreme waves are inherently nonlinear, for instance extreme Stokes waves are solutions of full Euler equations for inviscid,  incompressible fluids, and their existence (cf. \cite{AFT1982}) relies on the nonlinear character of these governing equations. Moreover, extreme waves incorporate numerous mathematical intricacies \cite{BT2003, Toland1996}, especially in regions close to the wave crest, which presents itself in the form of a cusp.  Nevertheless, numerous physical features of such flows have been analysed mathematically, for instance the particle trajectories and pressure in extreme Stokes waves in deep-water and water of finite depth may be found in  \cite{Constantin2012b,Lyons2014, Lyons2016a,Lyons2016b,Lyons2018}. In this work we will investigate internal geophysical flows of extreme form near the equator, and it will be found that the presence of a cusp at the wave crest of these flows once again introduces several mathematical difficulties. In particular, the necessary criteria to apply the implicit function theorem used to prove the existence of smooth geophysical flows are absent in the extreme case, while the cross-section of the wave about the equator also requires a more intricate analysis, both of which are addressed in this work.

The mathematical analysis of geophysical waves in the equatorial region has recently witnessed an explosion of research activity. A significant motivation for this interest has been the fact that such waves have an exact description in a variety of settings in the form of equatorially trapped  Gerstner like solutions, see \cite{Henry2018} for a review of three-dimensional Gerstner like solutions in the equatorial region.  Equatorially trapped surface waves have been investigated in \cite{Con2012,ConstantinJohnson2017, CM2017} while surface waves in the presence of underlying currents have been treated in \cite{Henry2013, HenrySastre2015}.  Internal geophysical flows are investigated in \cite{Constantin2013, CompelliIvanov2015,Matioc2013,Martin2015, RK2017, Klu2017}, while instabilities in such flows are addressed in \cite{CG2013, Ionescu2018,GD2014, HenryHsu2015}.

The waves we shall consider are inherently nonlinear and while this nonlinearity may preclude an exact analysis of their behaviour in many instances, this is not strictly the case, with the Gerstner solution \cite{Gerstner1809} being one of the few known exact solutions of the governing equations for incompressible, inviscid fluids in deep-water. The Gerstner wave was subsequently rediscovered by Rankine \cite{Rankine1863} and of particular interest it was extended to heterogeneous fluids by Dubreil-Jacotin \cite{DubreilJacotin1932}. A more modern perspective of Gerstner's wave is found in \cite{Constantin2001,Con2011,Henry2008} where the evolution of the fluid domain under the action of this solution is analysed using a combination of analytical and topological methods.

Of particular significance to the current work, Gerstner's solution may also be extended to hydrodynamic systems described in rotating coordinate systems. Gerstner-like solutions were initially used in \cite{Constantin2012a,Constantin2013} to model three-dimensional equatorially trapped geophysical flows and in \cite{Con2014} a Gerstner-like solution was constructed as an exact solution for internal geophysical water-waves near the equator in the $\beta$-plane formulation. This explicit solution is constructed in terms of Lagrangian variables (see \cite{Bennett2006} for a comprehensive review of the Lagrangian approach to fluid dynamics), and it is found that these waves may propagate in an eastward direction and are also found to decay significantly in the meridional direction away from the equator. Moreover, this Gerstner-like solution exists in a layer of finite depth, lying directly above a uniform current layer, which in turn lies above a transitional layer where the uniform current decays to zero when the motionless layer begins. The uniform current layer and the transitional layer are both subject to eastward currents, while it was shown in \cite{Con2014} that the Gersner-like solution induces a mean current towards the west, and so the model accounts for flow reversal, as observed in the Pacific equatorial undercurrent.

Notably, the results presented in \cite{Constantin2012a,Constantin2013, Con2014} concern nonlinear flows which are smooth throughout. In the current work we aim to extend the results presented in \cite{Con2014} to incorporate the case of nonlinear extreme waves in the Gerstner-like solution, which feature cusps at the wave-troughs along the equator in the thermocline, the interface between the Gerstner-like solution and the uniform current layer. The presence of these cusps leads to a number of mathematical complications, in particular it means the criteria necessary to confirm the existence of the extreme wave profile in the thermocline by means of the implicit function theorem are not satisfied along the equator. This issue is overcome in this work thus confirming the existence of extreme waves in  the thermocline. Based on this, we proceed to outline certain features  the meridional cross-section of the thermocline must display. Finally, we also confirm that current flow reversal across the thermocline  is preserved, even when it is disturbed to the form of an extreme Gerstner-like solution.

\section{Preliminaries}
\subsection{The governing equations in a rotating frame}
The rotation of the Earth about the polar axis induces a fictitious force, known as the Coriolis force, on all bodies away from the equator. In the following, the Earth is approximated to be an ideal sphere with radius $R=6731\, \mathrm{km}$ and rotating with angular velocity $\Omega=7.29\times10^{-5}\, \rm{rad\, s^{-1}}$ about the polar axis. The acceleration due to gravity is assumed uniform over Earth's surface with numerical value $g=9.8\, \rm{m\, s^{-2}}$.

\begin{figure}[h!]
\centering
\includegraphics[width=0.5\textwidth]{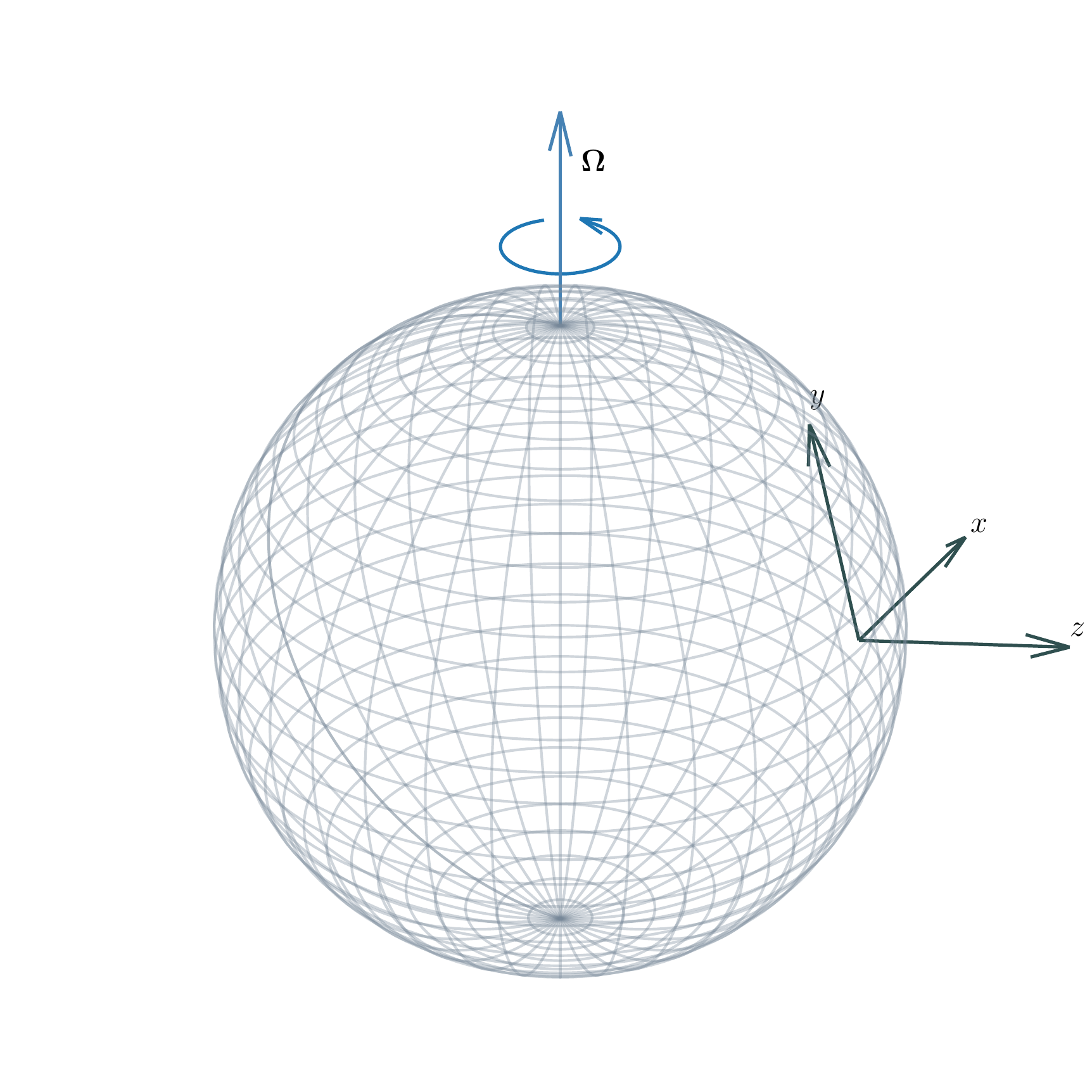}
\caption{The rotating $\left(x,y,z\right)$-coordinate system fixed to the surface of the  Earth. The $x$-axis points due-east, the $y$-axis points due-north and the $z$-axis points vertically upwards from the surface.}\label{fig1}
\end{figure}

In Figure \ref{fig1} we sketch the surface of the Earth with a Cartesian coordinate system, $(x,y,z)$ constructed at some point about the equator. The $x-$axis of this coordinate system points eastwards parallel to the equator, while the $y$-axis points due-north towards the pole. The $z$-axis is the line perpendicular to the surface of the Earth at this point. The latitude of the equator is set to be $0^{\circ}$, while the lines of latitude approximately $\pm5^{\circ}$ either side of the equator denote the boundary of the equatorial region.

In the following, we shall consider a fluid  in motion relative to the rotating $(x,y,z)$-coordinate axes with associated orthonormal basis vectors $\left\{\mathbf{e}_{x},\mathbf{e}_{y},\mathbf{e}_{z}\right\}$. This rotation is defined in relation to a reference frame which is stationary with respect to the centre of the Earth and whose orthonormal basis vectors we denote by $\left\{\mathbf{i},\mathbf{j},\mathbf{k}\right\}$. The coordinate transformation from the inertial to the rotating frame of reference is given by
\begin{equation}\label{s1eq1}
\begin{cases}
\mathbf{e}_{x}=-\sin\phi\mathbf{i}+\cos\phi\mathbf{j}\\
\mathbf{e}_{y}=-\sin\theta\cos\phi\mathbf{i}-\sin\theta\sin\phi\mathbf{j}+\cos\theta\mathbf{k}\\
\mathbf{e}_{z}=\cos\theta\cos\phi\mathbf{i}+\cos\theta\sin\phi\mathbf{j}+\sin\theta\mathbf{k}
\end{cases}
\end{equation}
with $\phi\in[0,2\pi]$ and $\theta\in\left[-\frac{\pi}{2},\frac{\pi}{2}\right]$ respectively denoting the zonal and meridional coordinates.

The inertial coordinate system may be chosen such that the rotation vector of the Earth is given by $\mathbf{\Omega}=\Omega\mathbf{k}$, where $\mathbf{k}$ is the unit vector projecting from the north pole along the polar axis. Inverting the transformation given by equation \eqref{s1eq1}, the rotation vector $\mathbf{\Omega}$ is given by
\begin{equation}\label{s1eq2}
\mathbf{\Omega}=\Omega\cos\theta\mathbf{e}_{y}+\Omega\sin\theta\mathbf{e}_{z}
\end{equation}
with respect to the rotating reference frame. Moreover, given any rotating orthonormal basis $\{\mathbf{e}_{x},\mathbf{e}_{y},\mathbf{e}_{z}\}$ with rotation vector $\mathbf{\Omega}$  it is always the case that the basis vectors satisfy
\begin{equation}\label{s1eq3}
\begin{aligned}
 \frac{d}{dt}{\mathbf{e}_{k}}=\mathbf{\Omega}\times\mathbf{e}_k,\quad \text{for }k\in\{x,y,z\}.
\end{aligned}
\end{equation}
and the reader is referred to \cite{Gol2002} for a comprehensive derivation of this formula. Thus a fluid particle following some trajectory
\[{\bf r}(t)=X(t)\mathbf{e}_{x}+Y(t)\mathbf{e}_{y}+Z(t)\mathbf{e}_{z},\]
will have a corresponding velocity and acceleration given by
\begin{equation}\label{s1eq4}
\mathbf{u}=\frac{\partial \mathbf{r}}{\partial t}+\mathbf{\Omega}\times\mathbf{r}\qquad\mathbf{a}=\frac{\partial \mathbf{u}}{\partial t}+2\mathbf{\Omega}\times\mathbf{u}+\mathbf{\Omega}\times\mathbf{\Omega}\times\mathbf{r}.
\end{equation}
The acceleration of a fluid element, $\mathbf{a}$, has two contributions due to fictitious forces  arising from the rotation of the Earth. The terms proportional to $\Omega^2$ correspond to the centripetal acceleration, which is the acceleration of a body away from the surface of the Earth due to the rotation of the surface. In practice, this term is negligible since it is effectively cancelled by the gravitational acceleration of the object, although, the effects of this force are observable in the deformation of the Earth away from a perfect sphere (see \cite{CRB2011} for further discussion). The term proportional to $\Omega$ is the result of the Coriolis force, and in the following we shall only consider the effect of this fictitious force.

\subsubsection{The Euler equation in rotating coordinates}
Given equation \eqref{s1eq4} and the effective cancellation of centripetal and gravitational forces on the surface of the Earth, it is found that the familiar material derivative is modified according to
\[\frac{\partial }{\partial t}+\mathbf{u}\cdot\nabla\to\frac{\partial }{\partial t}+\mathbf{u}\cdot\nabla+2\mathbf{\Omega}\times,\]
thus modifying the Euler equation according to
\begin{equation}\label{s1eq5}
\begin{cases}
 u_t+uu_x+vu_y+wu_z+2\Omega\cos\theta w-2\Omega\sin\theta v=-\frac{1}{\rho}P_{x}\\
 v_{t}+uv_{x}+vv_y+wv_z+2\Omega\sin\theta u=-\frac{1}{\rho}P_y\\
 w_{t}+uw_{x}+vw_y+ww_z-2\Omega\cos\theta u=-\frac{1}{\rho}P_z-g.
\end{cases}
\end{equation}
The hydrodynamic pressure in the fluid is represented by $P$ with the fluid density denoted by $\rho$, while $u$, $v$ and $w$ denote the velocity field components.

In the following, the meridional coordinate is restricted to the interval $-5^{\circ}<\theta<5^{\circ}$, and in this equatorial region  the approximation
\begin{equation*}\label{s1eq6}
\cos\theta\simeq 1\qquad \sin\theta\simeq\frac{y}{R}
\end{equation*}
may be used in equation \eqref{s1eq5}, cf. \cite{Con2012, Constantin2013,Con2014, ConstantinJohnson2017}. Introducing the parameter
\[\beta=\frac{2\Omega}{R}\]
this approximation yields the so-called $\beta$-plane formulation, given by
\begin{equation}\label{Euler-Beta}
\begin{cases}
  u_t+uu_x+wu_z+2\Omega w=-\frac{1}{\rho}P_{x}\\
  \beta y u=-\frac{1}{\rho}P_y\\
  w_{t}+uw_{x}+ww_z-2\Omega u=-\frac{1}{\rho}P_z-g,
\end{cases}
\end{equation}
where the velocity field is restricted by $v=0$ and $u_y=w_y=0$, meaning the equatorial flows under consideration are effectively restricted to zonal motion. While there is no motion in the meridional direction, this does not exclude variation in that direction, and in particular the depth of each interface between the layers of the flow will depend on latitude, thus making the flow three-dimensional. Incompressibility of the flow is ensured when
\begin{equation}\label{Incompressibility}
  u_{x}+v_{y}+w_{z}=0,
\end{equation}
while mass-conservation along the flow follows from
\begin{equation}\label{Continuity}
  \rho_{t}+\nabla\cdot(\mathbf{u}\rho)=0,
\end{equation}
and the system \eqref{Euler-Beta}--\eqref{Continuity} constitute the governing equations of the model. The boundary conditions coupled to these governing equations are of the form
\begin{equation}\label{BC}
\begin{aligned}
&\begin{rcases}
 P=P_0\\
 w=\eta_t+u\eta_x
\end{rcases}\quad\text{on }z=\eta(x,t)\\
&(u,w)=(0,0)\quad\text{for } z<-D+\frac{\kappa\beta}{4\Omega}y^2,
\end{aligned}
\end{equation}
where $\kappa>0$ is constant, and whose relevance will be clarified in section \ref{PhysicalFlow}. The first of these boundary conditions ensures the pressure on the free surface $z=\eta(x,t)$ is given by the constant atmospheric pressure $P_0$. The second boundary condition on this free surface implies that a fluid particle always remains on this surface. The final kinematic boundary condition ensures the fluid motion ceases below a certain depth, which demarcates the motionless layer.

\section{The flow about the equator}
The model resembles that presented in \cite{Con2014}, in that the waves occur in a fluid with two regions of different densities. The interface where the density changes is the thermocline, and it is this change in density which allows internal waves to arise. These internal waves propagate towards the east and move with constant wave-speed $c$ which will be prescribed by the dispersion relation in equation \eqref{wave-speed}. The waves, as already pointed out, have zero meridional  velocity, and occur beneath a surface layer $\mathcal{L}(t)$ whose behaviour is heavily influenced by wind effects. The behaviour of the fluid in this layer is not under consideration in this work. The {lower boundary} of the layer is described by the surface $z=\eta_{+}(x-ct,y)$, and beneath this boundary there are {four regions} wherein the fluid motion exhibits distinctive behaviour.

The upper most layer we investigate in this work is given by
\begin{equation}\label{Layer_M}
  \mathcal{M}(t)=\left\{(x,y,z):x\in\mathbb{R},y\in\mathbb{R},z\in\left(\eta_{0}(x-ct,y),\eta_{+}(x-ct,y)\right)\right\}.
\end{equation}
Wind effects on the internal flow in this work will be neglected for simplicity, however, it is known that atmospheric motion can play a role in the dynamics of internal flows in stratified layers, cf. \cite{ConstantinJohnson2017,Martin2015}.
The fluid density within this region is of constant value $\rho_{0}$, while the lower boundary of this region corresponds to the thermocline. Beneath the thermocline, the uniform flow layer is bounded from below by the surface
\[z=-d+\frac{\kappa\beta}{4\Omega}y^2\]
where $d$ is constant. In this region the fluid density increases slightly to $\rho_{+}>\rho_{0}$. The typical relative change in fluid density observed in the equatorial Pacific is of the order
\begin{equation}\label{relative_density}
\frac{\rho_{+}-\rho_{0}}{\rho_{0}}\approx 4\times 10^{-3},
\end{equation}
(see  \cite{KesslerMcPhaden1995}). Beneath this uniform flow layer, the transitional layer is bounded below by the surface
\[z=-D+\frac{\kappa\beta}{4\Omega}y^2,\]
wherein the fluid is again of constant uniform density $\rho_{+}$. The final layer, is unbounded below and corresponds to the motionless layer, wherein the fluid motion vanishes completely, and the fluid density is again $\rho_{+}$ throughout.
\begin{figure}[t!]
\includegraphics[width=\textwidth]{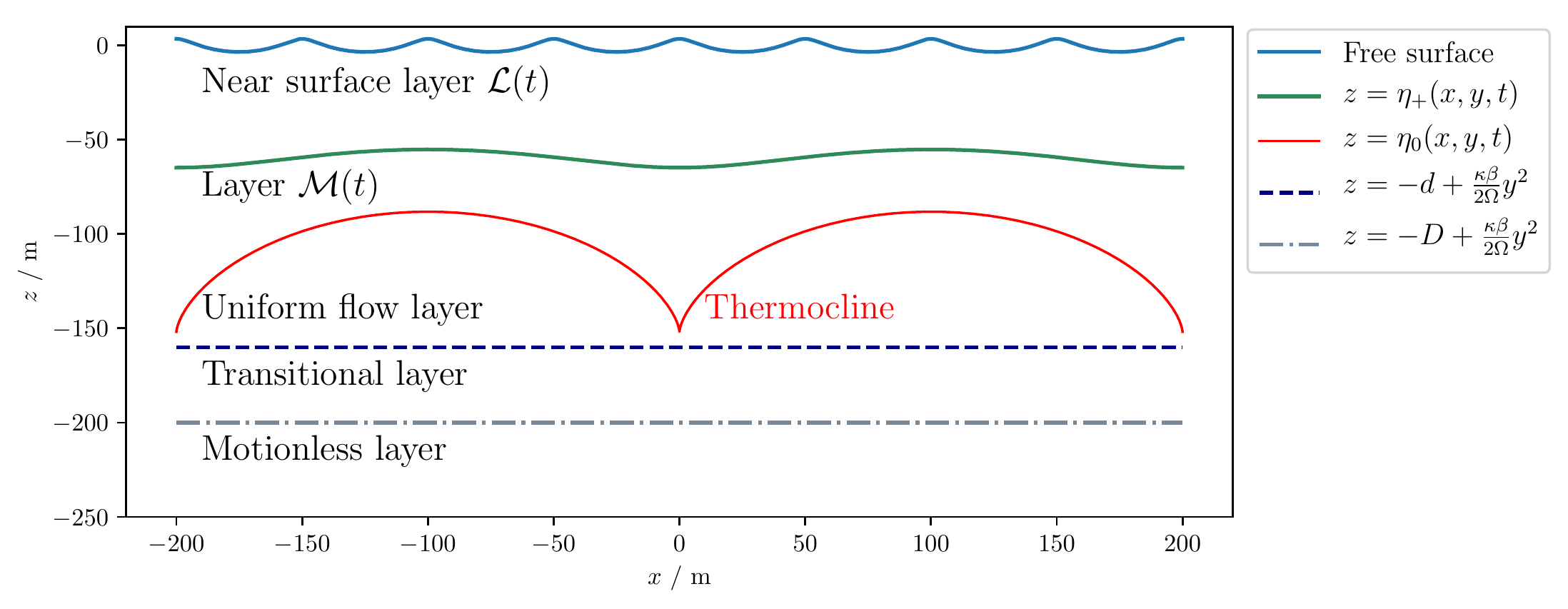}
\caption{A cross section of the flow in the equatorial plane $y=0$ with a flow wavelength of $L=200\,\mathrm{m}$. The average depth of the near surface layer is $60\,\mathrm{m}$, while the thermocline lies at an average depth of $120\,\mathrm{m}$. The transitional layer begins at $160\, \mathrm{m}$ approximately, while the motionless layer begins at $200\,\mathrm{m}$ beneath the free surface.}\label{layers}
\end{figure}

The layers of the fluid domain and the boundaries between those layers are shown in Figure \ref{layers}. In particular, the thermocline exhibits a cusp at the wave trough, which is one of the distinctive features of extreme waves.   In the case of extreme waves, the curve described by the thermocline is a cycloid, which is the path traced out by a particle at a fixed point on the circumference of a circle which rolls without slipping. In the case of smooth internal waves investigated in \cite{Con2014}, the thermocline wave profile is a trochoid, which is the path traced out by some point along the radius of a circle which rolls without slipping, cf. \cite{Con2011}.

\subsection{The layers beneath the thermocline}
The behaviour of the flow in the layers beneath the thermocline is identical to that in \cite{Con2014}, and we outline here the relevant features of these layers for the sake of completeness. The significant difference between the wave of extreme form and the internal waves previously investigated in that work occur in the layer $\mathcal{M}(t)$, in particular along the thermocline, due to the presence of cusps at the wave-trough as illustrated in Figure \ref{layers}.

\subsubsection{The motionless layer}
At any meridional location $y$ the upper boundary of the motionless layer is given by
\[z=-D+\frac{\kappa\beta}{4\Omega}y^2,\]
and beneath this boundary the fluid is still with $u=v=w=0$. In this case, the hydro-static pressure is given by
\begin{equation}\label{Pressure_D}
  P=P_0-\rho_+gz\quad\text{when }z\leq-D+\frac{\kappa\beta}{4\Omega}y^2
\end{equation}
The constant parameter $D$ is the depth of the upper boundary of this layer, as measured along the equator $y=0$.

\subsubsection{The transitional layer}
The upper boundary of this region is given by the surface
\[
z=-d+\frac{\kappa\beta}{4\Omega}y^2,
\]
where the constant parameter $d$ is the depth of the boundary as measured along the equator. Within this region the fluid is motionless in the vertical and meridional directions, however, the fluid does move in the zonal direction, with a velocity which depends on the depth and the meridional location. Specifically, the horizontal velocity in this layer is given by
\begin{equation}\label{s2eq3}
  \begin{aligned}
    u(x,y,z)=\frac{c}{D-d}\left(z-\frac{\kappa\beta}{4\Omega}y^2+D\right),\quad v=0,\quad w=0.
  \end{aligned}
\end{equation}
The vertical component of the Euler equation within this transitional layer now becomes
\begin{equation}\label{s2eq4}
  \frac{2\Omega c\rho_{+}}{D-d}\left(z-\frac{\kappa\beta}{4\Omega}y^2+D\right)=P_{z}+\rho_{+}g,
\end{equation}
and upon integrating we have
\begin{equation}\label{Pressure_C}
  \begin{aligned}
  P&=P\vert_{z=-D+\frac{\kappa\beta}{4\Omega}y^2}+\left[-\rho_{+}gz'+\frac{2\Omega c\rho_+}{D-d}\left(\frac{z'^2}{2}-\frac{\kappa\beta}{4\Omega}z'y^2+Dz'\right)\right]^{z}_{-D+\frac{\kappa\beta}{4\Omega}y^2}\\
  &=P_{0}-\rho_{+}gz+\frac{\Omega c\rho_{+}}{D-d}\left(z+D-\frac{\kappa\beta}{4\Omega}y^2\right)^2.
  \end{aligned}
\end{equation}
The pressure on the lower boundary of this region, $P\vert_{z=-D+\frac{\kappa\beta}{4\Omega}y^2}$, is deduced from the continuity of the pressure across the interface of two layers in the fluid domain, along with the expression \eqref{Pressure_D} for the pressure in the motionless layer.

\subsubsection{The uniform flow layer}
The upper boundary of the uniform flow layer corresponds to the thermocline
\[z=\eta_{0}(x-ct,y),\]
where $\eta_0$ is an eastward moving wave-form, moving with constant speed $c$ along any fixed meridional line $y$. Within the layer
\[-d+\frac{\kappa\beta}{4\Omega}y^2\leq z<\eta_0(x-ct,y)\]
the flow is uniform, given by
\begin{equation}\label{s2eq6}
 u=c,\quad v=0,\quad w=0.
\end{equation}
Replacing this velocity field in the Euler equation \eqref{Euler-Beta}, we deduce that
\begin{equation}\label{Pressure_B}
  P(x,y,z)=P_0-\rho_+gz+2\rho_{+}\Omega c\left(z+\frac{D+d}{2}-\frac{\kappa\beta}{4\Omega}y^2\right),
\end{equation}
having integrated each component and imposed equation \eqref{Pressure_C} on the lower boundary $z=-d+\frac{\kappa\beta}{4\Omega}y^2$.

\subsection{The layer above the thermocline}

The layer immediately above the thermocline, given by
\[\eta_{0}(x-ct,y)\leq z\leq \eta_{+}(x-ct,y)\]
is fully dynamic, in that the velocity field does not reduce to a simplified form as in all the layers beneath the thermocline. In this layer, it is convenient to analyse the flow in terms of Lagrangian coordinates, $(q,s,r)$, which serve to label the individual fluid particles within this layer. The flow in this layer is described in parametric form by
\begin{equation}\label{Gerstner}
\begin{cases}
  X(t;q,s,r)=q-\frac{1}{k}e^{-k\left(r+f(s)\right)}\sin(k(q-ct))\\
  Y(t;q,s,r)=s\\
  Z(t;q,s,r)=-d_{0}+r-\frac{1}{k}e^{-k(r+f(s))}\cos(k(q-ct)),
\end{cases}
\end{equation}
where  $L=\frac{2\pi}{k}$ is the wavelength of the flow, while the wave speed $c$ is deduced from the continuity of the pressure across the thermocline.
\begin{figure}[h!]
\centering
\includegraphics[width=\textwidth]{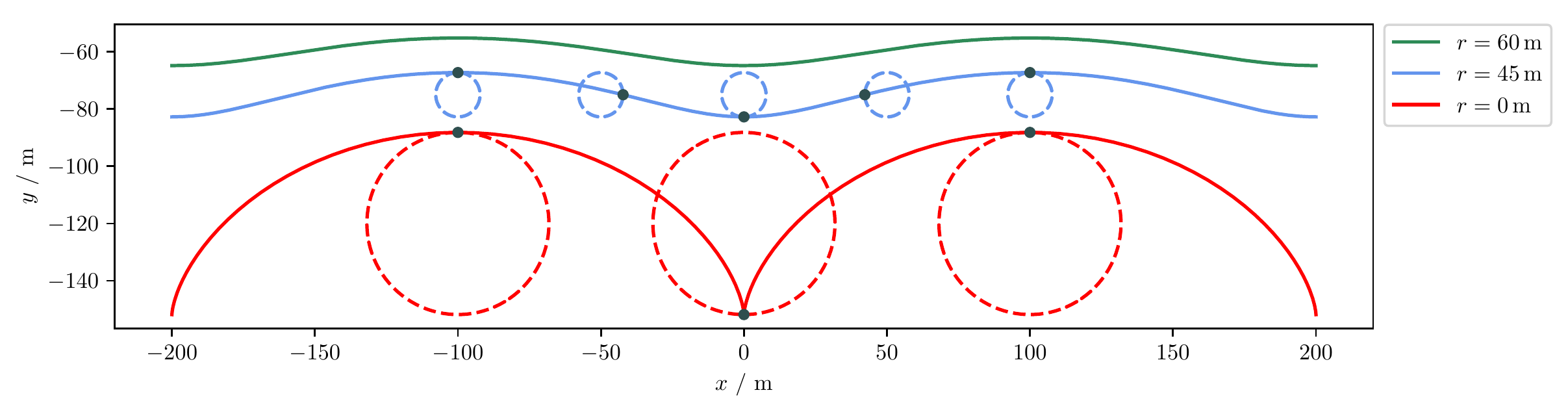}
\caption{The circular paths in the equatorial plane $y=0$ traced by particles in the layer $\mathcal{M}(t)$ are circles with centre $(q,r-d_0)$ whose radius increases with depth. The wavelength of the solution shown is $L=200\ \mathrm{m}$. The particle trajectories are counterclockwise and the particle positions shown in the figure are at time $t=0.$}\label{ParticlePaths}
\end{figure}
The function $f:=f(s)$ controls the decay of the fluid disturbance away from the equator which corresponds to the line $s=0$.
The parameter space variables
\[(q,s,r)\in\mathbb{R}\times(-s_0,s_0)\times(r_0(s),r_+(s))\]
may be interpreted as the initial locations of a motionless fluid body prior to being disturbed by a remotely generated wave propagating in the form \eqref{Gerstner}. The functions $r_0(s)$ and $r_{+}(s)$ characterise the surfaces $z=\eta_0(x,y,t)$ and $z=\eta_+(x,y,t)$ respectively, while $s_0\leq 250\, \mathrm{km}$ is the typical meridional extent of such disturbances (cf. \cite{Con2014}). Under the action of this flow, the particle labelled $(q,s,r)$ describes a circular path in the plane of fixed $s$, with centre $(q,r-d_0)$ and radius $\frac{1}{k}e^{-k(r+f(s))}$ which decreases with depth, as illustrated in Figure \ref{ParticlePaths}.

The dynamic possibility of the flow \eqref{Gerstner} has been confirmed in the context of surface waves in the work \cite{Sastre2015}, and in the context of internal geophysical waves in \cite{Rodriguez2017}. In either regime, as time advances, the mapping \eqref{Gerstner} remains a diffeomorphism from the parameter space $(q,s,r)$ to the fluid domain when $r\geq 0$. Meanwhile, the wave profile described by \eqref{Gerstner}, at any fixed $(s,r,t)$ is characterised to a large extent by the radius $e^{-k(r+f(s))}$. Crucially, for this profile to correspond to the graph of a function it is necessary that $r+f(s)\geq0$, with the wave of extreme form emerging when $r+f(s)=0.$ When $r+f(s)>0$ the wave-form is trochoidal, while the case $r+f(s)<0$ yields a curve which intersects with itself and so does not correspond to the graph of a function, cf. \cite{Con2011, Con2014}.

Introducing the parameters
\[\xi=k(r+f(s))\quad\theta=k(q-ct),\]
the Jacobian matrix of the coordinate transform \eqref{Gerstner} is of the form
\begin{equation}\label{Jacobian_Matrix}
\frac{\partial(X,Y,Z)}{\partial(q,s,r)}=
\begin{bmatrix}
1-e^{-\xi}\cos\theta&0&e^{-\xi}\sin\theta\\
f_{s}e^{-\xi}\sin\theta&1&f_{s}e^{-\xi}\cos\theta\\
e^{-\xi}\sin\theta&0&1+e^{-\xi}\cos\theta
\end{bmatrix}		
\end{equation}
while corresponding Jacobian is given by
\begin{equation}\label{Jacobian}
J=1-e^{-2\xi}.
\end{equation}
We note that this Jacobian is time-independent indicating that the flow is volume preserving. In the case of smooth waves studied in \cite{Con2014}, the Jacobian satisfies $J>0$ at all points in the closure of the layer $\mathcal{M}(t)$, and this allows the use of the implicit function theorem to infer the existence of an associated profile $r_0(s)$. In contrast, for extreme waves  this strict inequality now becomes $J\geq0$, with the Jacobian vanishing on the thermocline and specifically along the equator, which gives rise to a number of mathematical complications. In particular it precludes the use of the implicit function theorem to confirm the existence of waves of extreme form propagating in the thermocline, an issue which is addressed in the Section \ref{existence}.

A particularly appealing feature of the Lagrangian formalism is that the fluid velocity at any point along the flow given in equation \eqref{Gerstner} follows directly from
\begin{equation}\label{velocity}
\begin{cases}
  u=\frac{DX}{Dt}=ce^{-\xi}\cos\theta\\
  v=\frac{DY}{Dt}=0\\
  w=\frac{DZ}{Dt}=-ce^{-\xi}\sin\theta.
\end{cases}
\end{equation}
The corresponding acceleration along the flow is similarly obtained, yielding
\begin{equation}\label{acceleration}
\begin{cases}
  \frac{Du}{Dt}=kc^2e^{-\xi}\sin\theta\\
  \frac{Dv}{Dt}=0\\
  \frac{Dw}{Dt}=kc^2e^{-\xi}\cos\theta.
\end{cases}
\end{equation}
In terms of the labelling variables $(q,s,r)$, the gradient of the hydrodynamic pressure $P$ is obtained via
\begin{equation}\label{Prerssure_variable_change}
\nabla_{(q,s,r)}P=\frac{\partial(X,Y,Z)}{\partial(q,s,r)}\nabla_{(x,y,z)}P,
\end{equation}
and so applying equations \eqref{Jacobian_Matrix}--\eqref{Prerssure_variable_change} to the $\beta$-plane approximation \eqref{Euler-Beta}, we find
\begin{equation}\label{Pressu_label_gradient}
\begin{cases}
P_{q}=-\rho_0(kc^2-2\Omega c+g)e^{-\xi}\sin\theta\\
P_{s}=-\rho_{0}(kc^2-2\Omega c)f_{s}e^{-2\xi}-\rho_{0}(\beta cs+f_sg)e^{-\xi}\cos\theta\\
P_{r}=-\rho_{0}(kc^2-2\Omega c)e^{-2\xi}-\rho_{0}(kc^2-2\Omega c+g)e^{-\xi}\cos\theta-\rho_{0} g.
\end{cases}
\end{equation}
Integrating the expression for  $P_q$ and comparing with the expression for $P_{s}$, we deduce that
\begin{equation}\label{Decay_function}
f(s)=\frac{\beta s^2}{2(kc-2\Omega)}.
\end{equation}
Moreover the pressure distribution, written in terms of the labelling variables, is given by
\begin{equation}\label{Pressure_A}
  P(q-ct,s,r)=\rho_{0}\frac{kc^2-2\Omega c}{2k}e^{-2\xi}-\rho_{0}gr+\rho_{0}\frac{kc^2-2\Omega c+g}{k}e^{-\xi}\cos\theta + \tilde{P}_{0},
\end{equation}
where $\tilde{P}_0$ is a constant.

\subsubsection{The wave phase speed}
Since the thermocline $z=\eta_{0}(x,y,t)$ is identified with the parameterised surface $r=r_{0}(s)$, the pressure in the uniform current layer given by equation \eqref{Pressure_B} combined with the third component of equation \eqref{Gerstner} evaluated on $r_0(s)$, give the following expression for the pressure on the thermocline
\begin{equation}\label{Pressure1_tc}
P(q-ct,s,\eta_{0}) =P_0-\rho_+\frac{\kappa\beta c}{2}s^2+\rho_{+}(2\Omega c-g)\left(-d_0+r_0(s)-\frac{1}{k}e^{-\xi_{0}}\cos(\theta)\right)
\end{equation}
where $\xi_{0}=k(r_0(s)+f(s))$. Alternatively, the pressure in the layer $\mathcal{M}(t)$ given by \eqref{Pressure_A} and restricted to the thermocline $r=r_0(s)$ also yields
\begin{equation}\label{Pressure2_tc}
\begin{aligned}
P(q-ct,s,r_0(s))&=\tilde{P}_{0}+\rho_{0}\frac{kc^2-2\Omega c+g}{k}e^{-\xi_0}\cos\theta\\&\qquad+\rho_{0}\frac{kc^2-2\Omega c}{2k_0}e^{-2\xi}-\rho_{0}gr_0(s) .
\end{aligned}
\end{equation}
Since the pressure is continuous throughout the flow at all times, it follows from expressions \eqref{Pressure1_tc}--\eqref{Pressure2_tc} that
\begin{equation}\label{Pressure_transition_tc}
\rho_+(g-2\Omega c)=\rho_0(kc^2-2\Omega c+g),
\end{equation}
from which we conclude
\begin{equation}\label{wave-speed}
  c=\frac{\rho_{+}-\rho_{0}}{k\rho_{0}}\left(\sqrt{\Omega^2+kg\frac{\rho_{0}}{\rho_{+}-\rho_{0}}}-\Omega\right)>0,
\end{equation}
having retained only the positive root since the waves travel eastward. Crucially this dispersion relation highlights the necessity of a heterogeneous fluid body for the existence of these internal waves, since it is immediately clear that if $\rho_+=\rho_0$, then $c=0$ and the internal waves no longer exist. Decay of the internal wave with meridional distance from the equator is assured when $kc-2\Omega>0$, and so it follows from \eqref{wave-speed} that
\begin{equation}
    k>\frac{4\Omega^2}{g}\left(1+\frac{\rho_{0}}{\rho_{+}-\rho_{0}}\right),
\end{equation}
which combined with inequality \eqref{relative_density}, ensures that a maximal wavelength for internal waves is of order
$L=\frac{2\pi}{k}\approx1.2\times10^{7}\mathrm{m}. $

\subsection{Existence of extreme waves in the thermocline}\label{existence}
Equating the expressions given in equations \eqref{Pressure1_tc}--\eqref{Pressure2_tc} ensures the continuity of the pressure across the interface of the the layer $\mathcal{M}(t)$ and the uniform current layer. The equivalence of these two expression allows us to deduce an implicit relation for the function $r_0(s)$ of the form
\begin{multline}\label{Implicit1}
\rho_+\frac{\kappa\beta c}{2}s^2+\rho_0(kc^2-2\Omega c)\left(\frac{1}{2k}e^{-2k(r_0(s)+f(s))}+r_0(s)\right)\\=P_{0}-\tilde{P}_{0}+\rho_{+}gd_{0}-\rho_{+}\Omega c(2d_0-D-d).
\end{multline}
Let us denote
\[G(s.r)=\rho_+\frac{\kappa\beta c}{2}s^2+\rho_0(kc^2-2\Omega c)\left(\frac{1}{2k}e^{-2k(r+f(s))}+r\right),\]
and differentiating with respect to $r$ we find
\begin{equation}\label{Implicit-derivative}
 G_{r}(s,r)=\rho_{0}(kc^2-2\Omega c)(1-e^{-2k(r+f(s))}) > 0,\text{ when }r>0.
\end{equation}
In particular, smooth nonlinear flows always satisfy $r_0(s)\geq r_0^*>0$ for all $s\in(-s_0,s_0)$, and so under these circumstances it is clear that $r\mapsto G(s,r)$ is a strictly increasing diffeomorphism from the open interval $(0,\infty)$ onto the open interval $\left(\rho_+\frac{\kappa\beta c}{2}s^2+\rho_0\frac{(kc^2-2\Omega c)}{2k}e^{-2kf(s)},\infty\right)$. This in turn ensures the implicit function theorem may be used to infer the existence of a solution $r=r_0(s)$ of equation \eqref{Implicit1}. In addition we observer that $G(s,r)$ is effectively a function of $s^2$ only, and so the implicit function theorem ensures $r_0(s)$ is effectively a function of $s^2$ also, in which case it is clear that $r_{0}(s)$ is an even function of $s$, see \cite{Con2014} for further discussion.

Extreme waves emerge from \eqref{Gerstner} only when $r_0(s)\geq r_0^*=0$ with $r_0(s)=0$ along the equator $s=0$ (in contrast $r_0^*>0$ for smooth waves).  In this scenario equation \eqref{Implicit-derivative} now becomes $G_r(s,r)\geq0$ for $r\geq r_0$ and so significantly $G(s,r)$ is no longer a strictly increasing function of $r$ for all $s\in(-s_0,s_0)$ (see also the discussion after equation \eqref{Jacobian}). Equivalently, the fact that $G_r(s,r)$ is not strictly positive everywhere essentially means $G(s,r)$ is not automatically an injective function of $r$ for each fixed $s\in(-s_0,s_0)$, and so the implicit function theorem cannot be automatically applied to infer a solution of the relation \eqref{Implicit1}.
Nevertheless, even in the case of extreme waves it can still be shown for any fixed $s\in(-s_0,s_0)$ that the function
$r\mapsto G(s,r)$ is an injective map as follows:
\begin{lemma}\label{lem1}
Given any fixed $s$ the function
\[
\begin{aligned}
&G(s,\cdot):[0,\infty)\to\left[\rho_+\frac{\kappa\beta c}{2}s^2+\rho_0\frac{(kc^2-2\Omega c)}{2k}e^{-2kf(s)},\infty\right)\\
&G(s,\cdot):r\mapsto\rho_+\frac{\kappa\beta c}{2}s^2+\rho_0(kc^2-2\Omega c)\left(\frac{1}{2k}e^{-2k(r+f(s))}+r\right)
\end{aligned}
\]
is injective.
\end{lemma}
\begin{proof}
Away from the equator we see that
\begin{equation*}
 G_{r}(s,r)=\rho_{0}(kc^2-2\Omega c)(1-e^{-2k(r+f)}) >0,
\end{equation*}
since $\frac{\beta s^2}{2(kc-2\Omega)}>0$ when $s\neq 0$ and hence $r\mapsto G(s,r)$ is a strictly increasing function when $r\in[0,\infty)$.

When $s=0$ we note that $G(0,r)$ is continuous for any $r\in[0,\infty)$ and differentiable for any $r\in(0,\infty)$. Hence, the mean value theorem ensures that for any pair of values $0\leq r_1<r_2<\infty$, there exists some $\gamma\in(r_1,r_2)$ such that
\begin{equation}
\begin{aligned}
  \frac{G(0,r_2)-G(0,r_1)}{r_2-r_1}&=\rho_{0}(kc^2-2\Omega c)\left(1-e^{-2k(\gamma)}\right)>0
\end{aligned}
\end{equation}
since $0\leq r_1<\gamma $ and $kc^2-2\Omega c>0$. Thus $G(0,r)$ is injective for $r\in[0,\infty).$
It follows that the function $G(s,r)$ is an injective function of $r$ for every fixed $s$.
\end{proof}
Using the fact that $G(s,r)$ is an injective function on the interval $r\in[0,\infty)$ we may now show that there exists a solution of the implicit relation \eqref{Implicit1}.

\begin{theorem}\label{MainResult}
There exists an $s_0>0$ such that for  any $s\in(-s_0,s_0)$ the implicit relation
\begin{equation*}
G(s,r)=P_{0}-\tilde{P}_{0}+\rho_{+}gd_{0}-\rho_{+}\Omega c(2d_0-D-d)
\end{equation*}
has a unique solution $r=r_0(s)\geq 0$ such that $r(0)=0$. Moreover, $r_0(s)$ is a continuous function for $s\in(-s_0,s_0)$.
\end{theorem}
\begin{proof}
Firstly, we note since the right-hand side of equation \eqref{Implicit1} is constant and that $r_0(0)=0$, then equivalently we want to show there exists a continuous solution $r_0(s)$ to the implicit relation $G(s,r)=G(0,0)$.
Let us consider the behaviour of $G$ along the line segment
\[\left\{(0,r):0\leq r<\infty\right\}.\]
 \begin{figure}[h!]
	\centering
 	\includegraphics[width=0.75\textwidth]{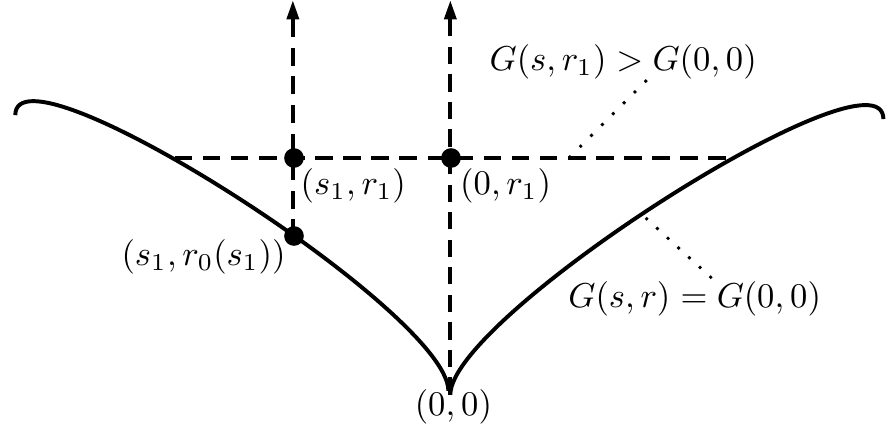}
 	\caption{A pictorial outline for the proof of Theorem \ref{MainResult}, confirming the existence of a solution of the implicit relation $G(s,r)=G(0,0)$ of extreme form.}\label{proof_figure}
 \end{figure}
It follows from Lemma \ref{lem1} that for any $r_1>0$ then $G(0,r_1)>G(0,0)$. Noting that $G(s,r)$ is continuous everywhere, it follows that an interval $(-s_0,s_0)$ may always be chosen so that
 \begin{equation}
   G(s,r_1)>G(0,0), \text{ for }s\in(-s_0,s_0),
 \end{equation}
see Figure \ref{proof_figure}. Now fixing some $s_1\in(-s_0,s_0)$ with $s_1\neq0$, we have
 \begin{equation}
  G(s_1,r_1)>G(0,0)
 \end{equation}
again. Moreover, since $G(s_1,r)$ is a continuous and injective function of $r$ (cf. Lemma \ref{lem1}), it follows that there is a unique $r_0(s_1)\in(0,r_1)$ such that
 \begin{equation}
  G(s_1,r_0(s_1))=G(0,0),
 \end{equation}
as illustrated in Figure \ref{proof_figure}. Hence, for all $s\in(-s_0,s_0)$, there exists a function $r_0(s)$ such that $G(s,r_0(s))=G(0,0)$.

The continuity of $r_0(s)$ may be seen to follow in a similar way. Fixing some $s_{1}\in(-s_{0},s_{0})$, to confirm the continuity we must show that for an arbitrarily small $\epsilon$ there is a a corresponding $\delta>0$ such that
\[\left\vert{r_{0}(s)-r_{0}(s_1)}\right\vert<\epsilon,\text{ for all }s\in(s_1-\delta,s_1+\delta).\]
However, we have already shown that an interval $(-s_0,s_0)$ may always be chosen which ensures the existence of some function $r(s)$ for each $s\in(-s_0,s_0)$. Consequently we may always choose an appropriate sub-interval $(s_1-\delta,s_1+\delta)\subset(-s_0,s_0)$ for each $s_1\in(-s_0,s_0)$, so that the continuity of $r(s)$ is ensured (see \cite{Cou2011} for further discussion).
\end{proof}

\section{Physical features of the flow}\label{PhysicalFlow}
\subsection{The profile near the equator}
Given any point $(s,r_0(s))$ where $G_{r}(s,r_0(s))\break>0$,  the derivative of $r_0(s)$ with respect to $s$ is given by
\begin{equation}\label{r_derivative}
 r^{\prime}(s)=-\left.\frac{G_{s}(s,r)}{G_r(s,r)}\right\vert_{r=r_0(s)}=-\frac{\beta cs(\rho_+\kappa-\rho_0 e^{-2\xi_0})}{\rho_0(kc^2-2\Omega c)(1-e^{-2\xi_0})}.
\end{equation}

\noindent We recall that Lemma \ref{lem1} and Theorem \ref{MainResult} require $r\in[0,\infty)$, and so it follows that $r_{0}(s)$ must not decrease, at least initially, as we move away from the equator. Hence, it follows from equation \eqref{r_derivative} that $r_0^{\prime}(s)\leq0$ for $s<0$ and $r_0^{\prime}(s)\geq0$ for $s>0$, at least about the equator.
This in turn requires
\begin{equation*}
 \rho_+\kappa-\rho_0e^{-2\xi_0}\leq0, \text{ for all } s\neq 0
\end{equation*}
from which it follows that
\begin{equation}\label{kappa_inequality}
\kappa<\frac{\rho_0}{\rho_+}<1.
\end{equation}
We note from equation \eqref{r_derivative} if $\kappa=\frac{\rho_0}{\rho_+}$ then $r_0^{\prime}(s)\leq0$ when  $s\geq0$ and vice-versa, meaning $r_0(s)$ decreases in moving away from the equator, hence the necessity for a strict inequality in \eqref{kappa_inequality}.

When $\kappa<\frac{\rho_0}{\rho_+}$, equation \eqref{Implicit1} has an exact solution of the form
\begin{equation}
r(s)=r_0^*+\frac{1}{2k}-\kappa\frac{\rho_+}{\rho_0}\frac{\beta}{(kc-2\Omega)}s^2+W\left(-e^{-1-2k(1-\kappa\frac{\rho_+}{\rho_0})\frac{\beta}{kc-2\Omega}s^2}\right),
\end{equation}
where $W(\cdot)$ is the Lambert $W$ function, defined by
\[z=W(z)e^{W(z)}\quad \text{ for all }z\in\mathbb{C},\]
The function $W(x)$ is real valued for any $x\in [-\frac{1}{e},\infty)$. The function is double valued on the interval $[-\frac{1}{e},0)$ with an increasing and decreasing branch in this interval, while the function is single valued on the interval $[0,\infty).$ An interesting application of the Lambert $W$ function for a convective Rayleigh-Benard model may be found in \cite{VPT2018}. In Figure \ref{Thermocline_graphs} in the left panel, the graph of $(s,r_0(s))$ is shown for a wave length $L=200\, $m and as expected the function is initially increasing when moving away from the equator. In the right panel, the cross section $(Y(q,s,r_(s),t),Z(q,s,r_0(s),t))$ is illustrated and again we see the shape of $r_0(s)$ reflected in this graph, except at the wave crest where the profile becomes smooth. In Figure \ref{Thermocline_surface} the thermocline is illustrated from beneath to emphasise the cusp at the wave trough, one of the hallmarks of extreme waves. We also note from the figures that the wave profile does not appear to be differentiable at the wave trough, again a common feature of extreme waves.
\begin{figure}[h!]
\includegraphics[width=0.5\textwidth]{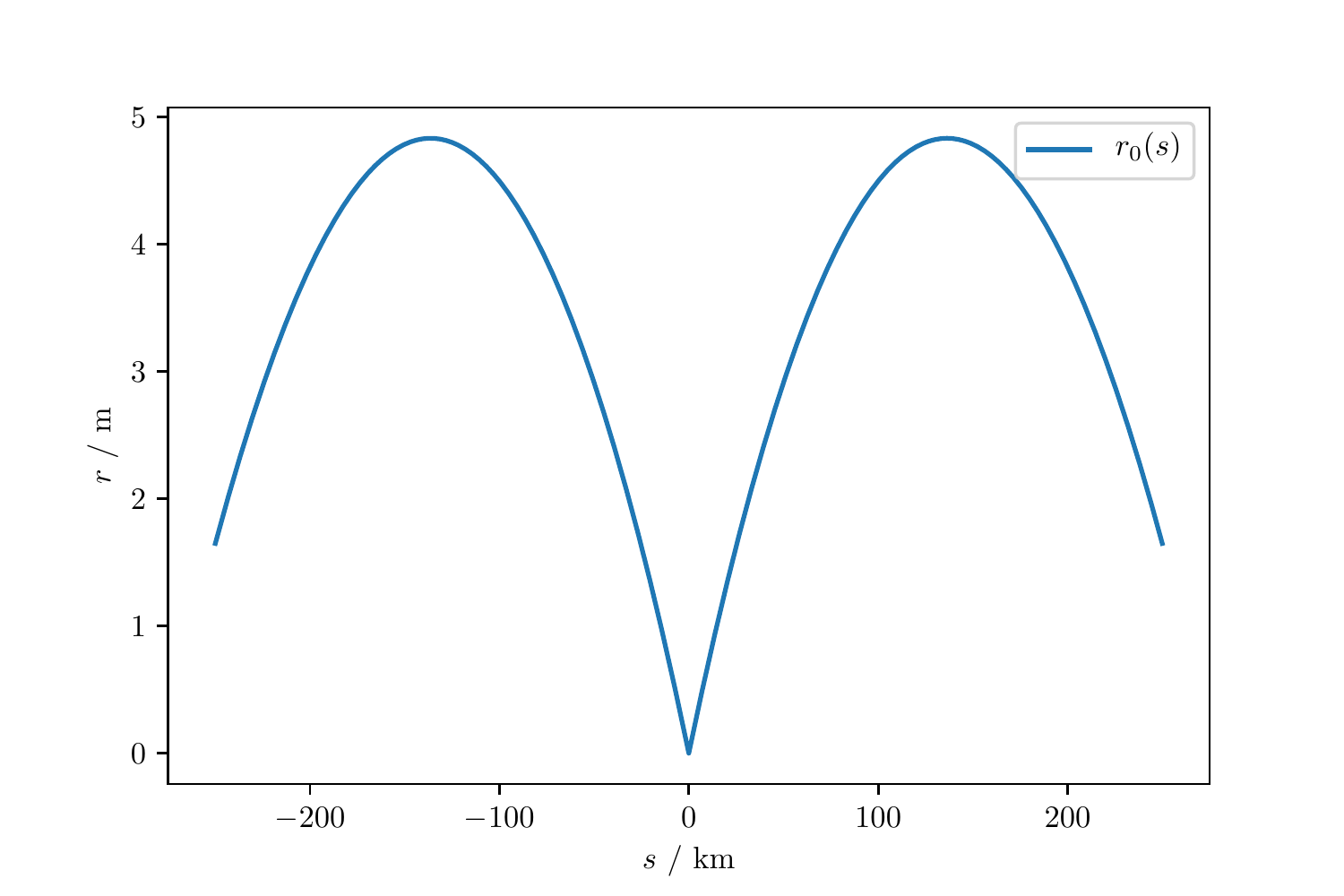}\includegraphics[width=0.5\textwidth]{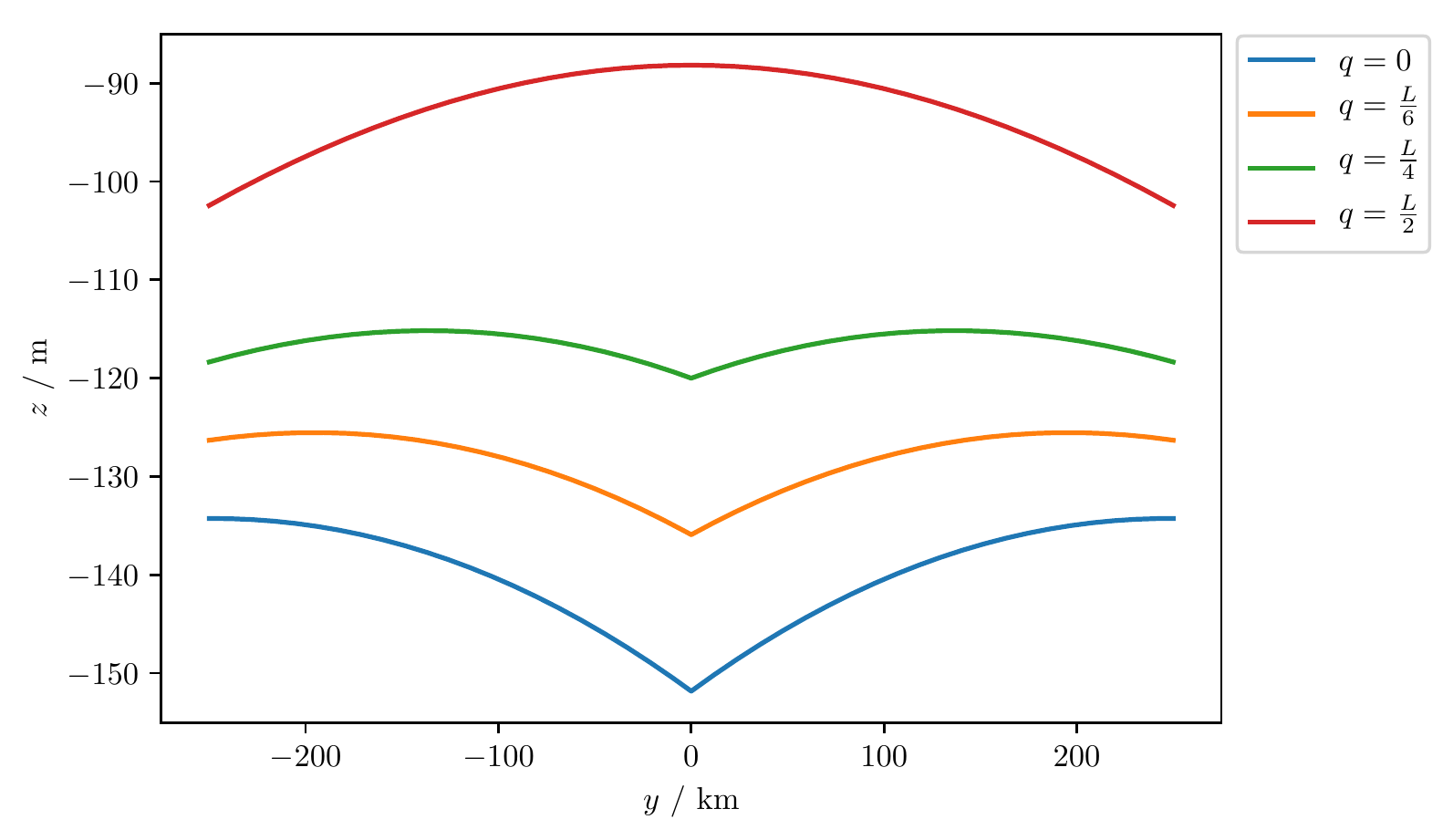}
\caption{The left-hand panel shows the graph $(s,r_0(s))$ when $\kappa=0.5$ and $L=200$ meters. The parameter $d_0=120$ ensures a mean equatorial depth of $120\, $m. In the right-hand panel the corresponding thermocline profile in the fluid domain, at the equatorial locations indicated in the figure.}\label{Thermocline_graphs}
\end{figure}
\begin{figure}[h!]
\centering
\includegraphics[width=\textwidth]{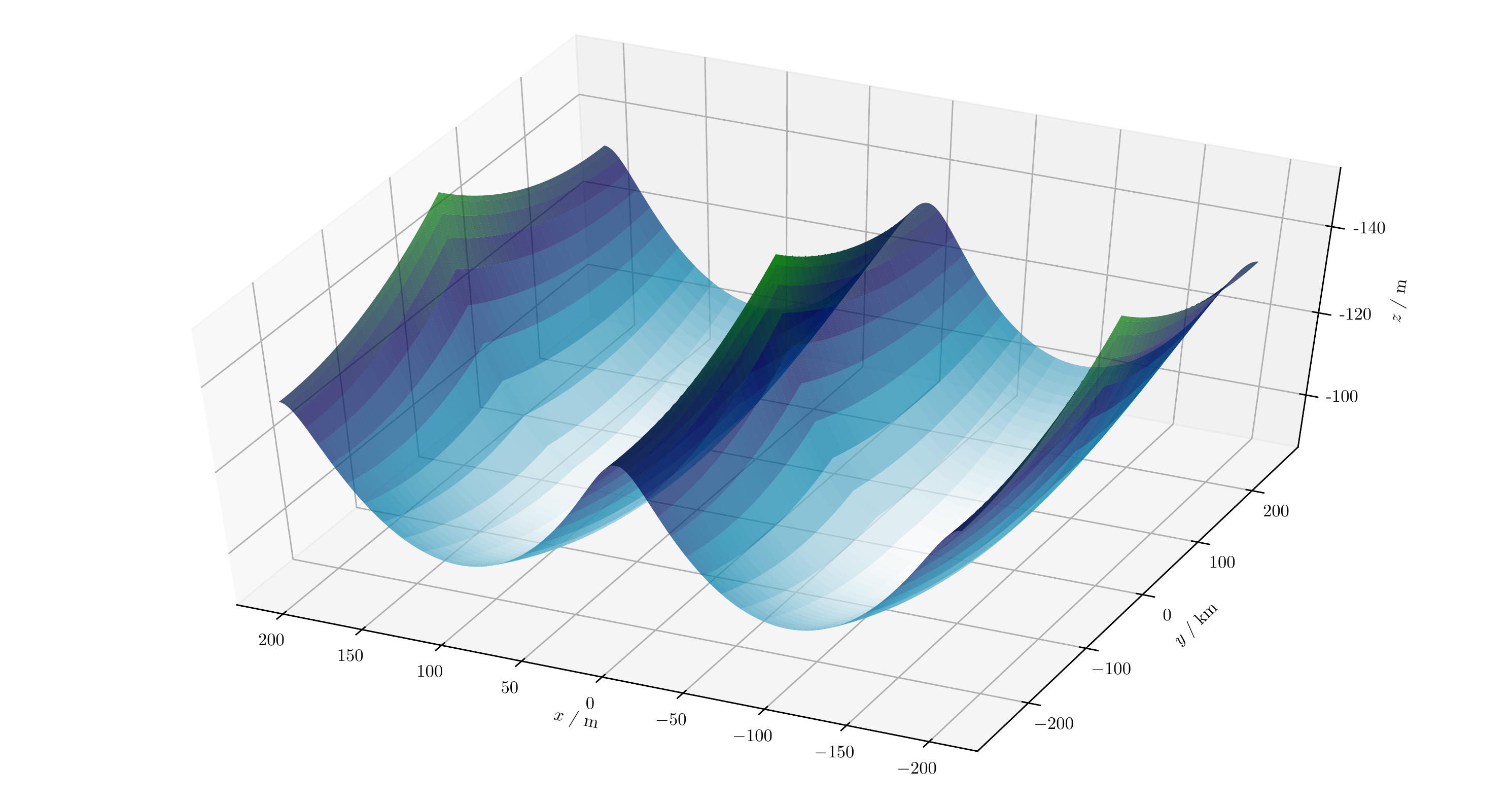}
\caption{The thermocline surface when $L=200\, $m, $r_0^*=0\, $m and $d_0=120\, $m, with $\kappa=0.5$. The surface is viewed from beneath and the value of $\kappa$ has been chosen to emphasise the extreme behaviour of the wave along the equator at the wave-troughs.}\label{Thermocline_surface}
\end{figure}

When the wavelength $L$ is sufficiently small, there are values $s\in(-s_0,s_0)$ such that the solution $r_0(s)<0$, for instance when $L=100\, $m we find that $r_0(s)<0$ when $\vert s\vert \gtrsim 234\,$km. However, we also note that for $|s|\gtrsim234\, $km then $f(s)+r(s)>0$, meaning the details of Lemma \ref{lem1} and Theorem \ref{MainResult} remain valid, and crucially the graph of $(X(q,r_0(s),s,t),Z(q,r_0(s),s,t))$ does not intersect itself. Thus, it is only near the equator that we must insist $r_0(s)>0$, in line with our considerations above.

\subsection{Current flow reversal}
In \cite{Con2014} it was shown that the mean Eulerian velocity of the solution \eqref{Gerstner} induces a westward flowing current in the near surface region $M(t)$, and owing to the fact that current in the layer beneath the thermocline is eastward, it was found that such behaviour reflects the features of the equatorial undercurrent (with relevant data in \cite{Wac1988}). We shall now demonstrate that this property of the dynamic solution is also a feature of the flow in the extreme case, thus maintaining one of the most attractive features of the flow from a geophysical perspective.

At any given latitude $s\in(-s_0,s_0)$, any fixed depth $z=z_0$ within the layer $M(t)$ may be written in terms of Lagrangian variables via the solution \eqref{Gerstner}, in the form
\begin{equation}\label{constant-depth}
 z_0=-d_0+r+\frac{1}{k}e^{-k(r+f(s))}\cos(k(q-ct)),
\end{equation}
whose solution is defined implicitly via
\begin{equation}\label{implicit-relation2}
 r=\zeta(q,s,z_0,t).
\end{equation}
Since the thermocline $r(s)=r_0^{*}$ is the lower boundary of this layer, it follows that $\zeta(q,s,z_0,t)\geq r_{0}^{*}$ for all $q,s,t$ and $z\geq\eta_0(x,y,t)$. Replacing \eqref{implicit-relation2} in \eqref{constant-depth} and differentiating with respect to $q$ we deduce
\begin{equation}\label{alpha_q}
\zeta_q=-\frac{e^{-\xi}\sin\theta}{1+e^{-\xi}\cos\theta}.
\end{equation}
The mean horizontal flow at this latitude $s$ and depth $z_0$ is then given by
\begin{equation}
 c+\bar{u}(s,z_0)=\frac{1}{T}\int_0^{T}\left[c+u(x-ct,s,z_0)\right]dt.
\end{equation}
Using
\[t=\frac{L}{c}\qquad dx=cdt\]
along with the change of variables (having fixed $s$ and $z$)
\[dx=\frac{\partial x}{\partial q}dq,\]
it follows from \eqref{Gerstner}, \eqref{velocity} and \eqref{alpha_q} that the mean horizontal current is given by the integral
\begin{equation}
 \bar{u}(s,z_0)=-\frac{c}{L}\int_{0}^{L}e^{-2\xi}dq \Rightarrow -c<\bar{u}(s,z_0)<0,
\end{equation}
meaning there is a net westward flow in the layer $M(t)$. Thus there is a flow reversal in the transition across the thermocline when the profile is smooth or of extreme form.

\section*{Acknowledgments}
The author is grateful to the referees for several helpful comments.
The author is grateful to the organisers of the programme ``Mathematical Aspects of Physical Oceanography (January 2019 -- March 2019)'' at the Erwin Schr\"{o}dinger Institute, Vienna, Austria.

\medskip
Received August 2018; revised October 2018.
\medskip


\begin{thebibliography}{10}

\bibitem{AFT1982}
\newblock C.~J. Amick, L.~E. Fraenkel and J.~F. Toland,
\newblock {On the {S}tokes conjecture for the wave of extreme form},
\newblock \emph{Acta Math.}, \textbf{148} (1982), 193--214.

\bibitem{Bennett2006}
\newblock A.~Bennett,
\newblock \emph{Lagrangian Fluid Dynamics},
\newblock Cambridge Monographs on Mechanics, Cambridge University Press, 2006.

\bibitem{BT2003}
\newblock B.~Buffoni and J.~Toland,
\newblock \emph{Analytic Theory of Global Bifurcation},
\newblock Princeton University Press, 2003.

\bibitem{CompelliIvanov2015}
\newblock A.~Compelli and R.~Ivanov,
\newblock {On the dynamics of internal waves interacting with the equatorial undercurrent},
\newblock \emph{J. Nonlin. Math. Phys.}, \textbf{22} (2015), 531--539.

\bibitem{Constantin2001}
\newblock A.~Constantin,
\newblock {On the deep water wave motion},
\newblock \emph{J. Phys. A}, \textbf{34} (2001), 1405--1417.

\bibitem{Con2011}
\newblock A.~Constantin,
\newblock \emph{Nonlinear Water Waves with Applications to Wave-Current Interactions and Tsunamis}, vol.~81 of CBMS-NSF Regional Conference Series in Applied Mathematics,
\newblock SIAM, Philadelphia, 2011.

\bibitem{Con2012} 
\newblock A.~Constantin,
\newblock {An exact solution for equatorially trapped waves},
\newblock \emph{J. Geophys. Res. Oceans}, \textbf{117} (2012).

\bibitem{Constantin2012a} 
\newblock A.~Constantin,
\newblock {On the modelling of equatorial waves},
\newblock \emph{Geophys. Res. Lett.}, \textbf{39} (2012).

\bibitem{Constantin2012b}
\newblock A.~Constantin,
\newblock {Particle trajectories in extreme {S}tokes waves},
\newblock \emph{IMA J. Appl. Math.}, \textbf{77} (2012), 293--307.

\bibitem{Constantin2013} 
\newblock A.~Constantin,
\newblock {Some three-dimensional nonlinear equatorial flows},
\newblock \emph{J. Phys. Oceanog.}, \textbf{43} (2013), 165--175.

\bibitem{Con2014}
\newblock A.~Constantin,
\newblock {Some nonlinear, equatorially trapped, nonhydrostatic internal geophysical waves},
\newblock \emph{J. Phys. Oceanogr.}, \textbf{44} (2014), 781--789.

\bibitem{CG2013}
\newblock A.~Constantin and P.~Germain,
\newblock {Instability of some equatorially trapped waves},
\newblock \emph{J. Geophys. Res. Oceans}, \textbf{118} (2013), 2802--2810.

\bibitem{ConstantinJohnson2017}
\newblock A.~Constantin and R.~S. Johnson,
\newblock A nonlinear, three-dimensional model for ocean flows, motivated by some observations of the {P}acific {E}quatorial {U}ndercurrent and thermocline,
\newblock \emph{Phys. Fluids}, \textbf{29} (2017), 056604.

\bibitem{CM2017}
\newblock A.~Constantin and S.~G. Monismith,
\newblock {Gerstner waves in the presence of mean currents and rotation},
\newblock \emph{J. Fluid Mech.}, \textbf{820} (2017), 511--528.

\bibitem{Cou2011}
\newblock R.~Courant,
\newblock \emph{Differential and Integral Calculus},
\newblock John Wiley \& Sons, Inc., New York, 1988.

\bibitem{CRB2011}
\newblock B.~Cushman-Roisin and J.~M. Beckers,
\newblock \emph{Introduction to {G}eophysical {F}luid {D}ynamics: {P}hysical and {N}umerical {A}spects}, vol. 101,
\newblock Academic Press, 2011.

\bibitem{DubreilJacotin1932}
\newblock M.~L. Dubreil-Jacotin,
\newblock Sur les ondes de type permanent dans les liquides,
\newblock \emph{Atti. Accad. Naz. Lincei, Mem. Cl. Sci. Fis. Mat. Nat.}, \textbf{15} (1932), 814--819.

\bibitem{GD2014}
\newblock F.~Genoud and D.~Henry,
\newblock {Instability of equatorial water waves with an underlying current},
\newblock \emph{J. Math. Fluid Mech.}, \textbf{16} (2014), 661--667.

\bibitem{Gerstner1809}
\newblock F.~Gerstner,
\newblock Theorie der wellen samt einer daraus abgeleiteten theorie der deichprofile,
\newblock \emph{Ann. Phys.}, \textbf{2} (1809), 412--445.

\bibitem{Gol2002}
\newblock H.~Goldstein, C.~P. Poole and J.~L. Safko,
\newblock \emph{Classical {M}echanics},
\newblock Pearson International Edition, Addison Wesley, 2002.

\bibitem{Henry2008}
\newblock D.~Henry,
\newblock {On {G}erstner's water wave},
\newblock \emph{J. Nonlin. Math. Phys.}, \textbf{15} (2008), 87--95.

\bibitem{Henry2013}
\newblock D.~Henry,
\newblock {An exact solution for equatorial geophysical water waves with an underlying current},
\newblock \emph{Eur. J. Mech. B Fluids}, \textbf{38} (2013), 18--21.

\bibitem{Henry2018} 
\newblock D.~Henry,
\newblock {On three-dimensional {G}erstner-like equatorial water waves},
\newblock \emph{Phil. Trans. R. Soc. A}, \textbf{376} (2018), 20170088, 16pp.

\bibitem{HenryHsu2015}
\newblock D.~Henry and H.~C. Hsu,
\newblock {Instability of internal equatorial water waves},
\newblock \emph{J. Diff. Eqn.}, \textbf{258} (2015), 1015--1024.

\bibitem{HenrySastre2015}
\newblock D.~Henry and S.~Sastre-Gomez,
\newblock {Mean flow velocities and mass transport for equatorially-trapped water waves with an underlying current},
\newblock \emph{J. Math. Fluid Mech.}, \textbf{18} (2016), 795--804.

\bibitem{Ionescu2018}
\newblock D.~Ionescu-Kruse,
\newblock {On the short-wavelength stabilities of some geophysical flows},
\newblock \emph{Phil. Trans. R. Soc. A}, \textbf{376} (2018), 20170090, 21pp.

\bibitem{Klu2017}
\newblock M.~Kluczek,
\newblock {Exact and explicit internal equatorial water waves with underlying currents},
\newblock \emph{J. Math. Fluid Mech.}, \textbf{19} (2017), 305--314.

\bibitem{Lyons2014}
\newblock T.~Lyons,
\newblock {Particle trajectories in extreme {S}tokes waves over infinite depth},
\newblock \emph{Disc. Contin. Dyn. Sys. Ser. A}, \textbf{34} (2014), 3095--3107.

\bibitem{Lyons2016a}
\newblock T.~Lyons,
\newblock {The pressure distribution in extreme {S}tokes waves},
\newblock \emph{Nonlin. Anal. Real World Appl.}, \textbf{31} (2016), 77--87.

\bibitem{Lyons2016b}
\newblock T.~Lyons,
\newblock {The pressure in a deep-water {S}tokes wave of greatest height},
\newblock \emph{J. {M}ath. Fluid Mech.}, \textbf{18} (2016), 209--218.

\bibitem{Lyons2018}
\newblock T.~Lyons,
\newblock {The dynamic pressure in deep-water extreme {S}tokes waves},
\newblock \emph{Phil. Trans. R. Soc. A}, \textbf{376} (2018), 20170095, 13pp.

\bibitem{Martin2015}
\newblock C.~I. Martin,
\newblock {Dynamics of the thermocline in the equatorial region of the pacific ocean},
\newblock \emph{J. Nonlin. Math. Phys.}, \textbf{22} (2015), 516--522.

\bibitem{Matioc2013}
\newblock A.~V. Matioc,
\newblock {Exact geophysical waves in stratified fluids},
\newblock \emph{Appl, Anal.}, \textbf{92} (2013), 2254--2261.

\bibitem{Rankine1863}
\newblock W.~J.~M. Rankine,
\newblock On the exact form of waves near the surface of deep water,
\newblock \emph{Philos. Trans. R. Soc. London A}, \textbf{153} (1863), 127--138.

\bibitem{Rodriguez2017}
\newblock A.~Rodriguez-Sanjurjo,
\newblock {Global diffeomorphism of the {L}agrangian flow-map for {E}quatorially-trapped internal water waves},
\newblock \emph{Nonlin. Anal.: Theor., Meth. Appl.}, \textbf{149} (2017), 156--164.

\bibitem{RK2017}
\newblock A.~Rodriguez-Sanjurjo and M.~Kluczek,
\newblock {Mean flow properties for equatorially trapped internal water wave--current interactions},
\newblock \emph{Appl. Anal.}, \textbf{96} (2017), 2333--2345.

\bibitem{KesslerMcPhaden1995}
\newblock K.~W. S and M.~M. J,
\newblock Oceaninc equatorial waves and the 1991--1993 {E}l {N}i\~{n}o,
\newblock \emph{J. Climate}, \textbf{8} (1995), 1757--1774.

\bibitem{Sastre2015}
\newblock S.~Sastre-Gomez,
\newblock {Global diffeomorphism of the {L}agrangian flow-map defining equatorially trapped water waves},
\newblock \emph{Nonlin. Anal.}, \textbf{125} (2015), 725--731.

\bibitem{Stokes1880}
\newblock G.~G. Stokes,
\newblock Considerations relative to the greatest height of oscillatory irrotational waves which can be propagated without change of form,
\newblock \emph{Mathematical and Physical Papers}, \textbf{1} (1880), 225--228.

\bibitem{Toland1996}
\newblock J.~F. Toland,
\newblock {Stokes waves},
\newblock \emph{Topol. Meth. Nonlin. Anal.}, \textbf{7} (1996), 1--48.

\bibitem{VPT2018}
\newblock G.~K. Vallis, D.~J. Parker and S.~M. Tobias,
\newblock {A simple system for moist convection: The rainy-{B}enard model},
\newblock \emph{J. Fluid Mech.}, \textbf{862} (2019), 162--199.

\bibitem{Wac1988}
\newblock S.~H.~C. Wacogne,
\newblock \emph{Dynamics of the Equatorial Undercurrent and Its Termination},
\newblock PhD thesis, Massachusetts Institute of Technology, 1988.


\end{thebibliography}
\end{document}